%% file: root.tex
\documentclass[letterpaper, 10 pt, conference]{ieeeconf}
\input{Packages}

\input{CustomCommands}

\title{\LARGE \bf
Data-Driven Personalization of Automated Insulin Delivery}

\author{Ali Kashani$^{1}$, Ali Tavasoli$^{2}$, and Heman Shakeri$^{1}$% <-this % stops a space
}

\begin{document}
\bstctlcite{BSTcontrol}
\maketitle
\thispagestyle{empty}
\pagestyle{empty}
\begingroup
\renewcommand\thefootnote{}
\footnotetext{
$^{1}$Ali Kashani ({fww9ba@virginia.edu}), and Heman Shakeri ({hs9hd@virginia.edu}), are with the School of Data Science, University of Virginia, Charlottesville, VA 22903.

$^{2}$Ali Tavasoli ({hqyq4k@jmu.edu}), is with the Department of Mathematics and Statistics, James Madison University, Harrisonburg, VA 22807.

The authors gratefully acknowledge the support of the LaunchPad for Diabetes program at the University of Virginia.}
\endgroup
\begin{abstract}
Automated insulin delivery (AID) systems are often tuned for the population and offer limited online adaptation to the inter- and intrapatient variability in insulin needs caused by meal patterns, physical activity, and fluctuations in insulin sensitivity. We present a real-time, data-driven personalization approach that adapts controller parameters using the subject's daily glycemic data. The adaptation is formulated as projected gradient descent on a daily risk metric, where the gradient estimation is designed to attenuate noise and metabolic variability. We use contraction theory to validate the optimization framework and convergence of the closed-loop system under adaptation. \textit{In silico} experiments on the $100$-adult cohort of the FDA-accepted UVA/Padova T1D simulator show that our method improves glycemic risk and increases time-in-range (TIR, $70$--$180$\,mg/dL) by $2$\%, $3$\%, and $4$\% after $4$, $8$, and $17$ weeks, respectively, under variability in meal timing, meal size, and insulin sensitivity.
\end{abstract}

\section{Introduction}
\Ac{t1d} is an autoimmune disease causing the loss of endogenous insulin secretion. Therefore, exogenous insulin should be delivered to reduce average glucose and postprandial excursions while avoiding hypoglycemia. \Ac{aid} systems are an integration of \ac{cgm}, insulin infusion, and feedback control. Among the available control strategies, \ac{mpc} is especially attractive because it can encode clinical objectives, such as target glucose zones, insulin-on-board limits, and asymmetric hypo/hyperglycemia penalties, directly in a receding-horizon optimization problem~\cite{gondhalekar2016periodic, shi2018adaptive,  moscoso2023reshaping}.

Despite advances, most current \ac{aid} controllers are tuned at the population level. A common approach is optimizing controller parameters by exhaustive search over a cohort of virtual subjects in the FDA-approved UVA/Padova \ac{t1d} simulator~\cite{ shi2018adaptive, Kovatchev2009InSilicoT1D}. However, population-based tuning prioritizes safety across the entire population, which may be conservative for individuals whose glycemic profile differs from the average population~\cite{Magni2007RunToRunMPC}. This is important in \ac{t1d}, since insulin needs vary across subjects and within the same subject over time due to meal composition, physical activity, stress, circadian effects, and fluctuations in insulin sensitivity~\cite{hinshaw2013diurnal, wiegand2008daily}. Diurnal insulin sensitivity patterns are known to be strongly individual-specific and difficult to extrapolate across the population~\cite{hinshaw2013diurnal}. These considerations motivate online personalization mechanisms that adapt controller behavior directly from an individual's online data.

In~\cite{pryor2025pogo}, an individual's \ac{tdi} is estimated from their history and used to scale the controller. In~\cite{moscoso2023performance, villa2022model}, model-based adaptive schemes are employed to re-estimate model parameters, such as insulin sensitivity, used within the \ac{mpc}. While effective, such approaches depend on the identifiability and fidelity of the underlying physiological model. An alternative is direct adaptation, in which controller parameters are tuned from glycemic outcomes, without re-identifying physiology. One example of direct adaptation is \ac{esc}, where controller parameters are perturbed, the gradient of a risk metric is estimated through modulation/demodulation, and the parameters are updated by gradient descent~\cite{cao2017extremum}. However, in \ac{aid} applications, the combination of \ac{cgm} noise, metabolic and behavioral variability, and changing operating conditions makes gradient estimation challenging. In standard \ac{esc}, improving robustness requires stronger temporal low-pass filtering, which can slow convergence and introduce oscillations as the operating point drifts.

In this work, we present a real-time, performance-based adaptation scheme for \ac{aid} that personalizes controller parameters using the individual's glycemic data. Our approach combines projected gradient descent with a \ac{rwls} gradient estimator. The estimator emphasizes past observations local to the current value and regularizes updates based on the level of locality, thereby improving robustness to day-to-day metabolic variability and noise. We prove that the resulting estimation error decomposes into a relative bias that scales with the persistent excitation condition, noise, and metabolic variability.

% A central analytical challenge is that the daily glycemic risk depends on the full closed-loop glucose trajectory, not only on the controller parameters. To address this, we use contraction theory to show that, under a nominal daily routine, the closed-loop glucose dynamics entrain to a unique parameterized periodic orbit. This induces a well-defined steady-state daily cost $\bar J(\theta)$ and reduces the run-to-run personalization problem to optimization over the controller parameter space. The transient mismatch between the realized daily cost and this steady-state cost is then treated separately in the stability analysis.
 
The main contributions of this paper are as follows:
\begin{enumerate}
    \item[(i)] We develop an \ac{rwls} gradient estimator with spatial proximity weighting and adaptive regularization for \ac{aid} personalization. The \ac{rwls} gradient estimator admits an error bound which is controlled under curvature regularity.
    
    \item[(ii)] We show, using contraction theory, that under a nominal daily routine the closed-loop glucose dynamics reduce to a periodic orbit. This yields a well-defined steady-state risk $\bar J(\theta)$ and provides a basis for analyzing the adaptation law as projected gradient descent.
    
    \item[(iii)] We show the system closed-loop with \ac{pas} is \ac{iss}, and validate it through \textit{in silico} trials on the FDA-approved UVA/Padova \ac{t1d} simulator under realistic meal, behavioral, and metabolic variability.
\end{enumerate}

\subsection*{Notation and definitions:}
Let \(\reals_+ := \{x\in\reals: x\ge 0\}\) and $\naturals_a^b = \{i \in \naturals : a \le i \le b\}$, where \(\reals\) and $\naturals$ denote real and integer numbers, respectively. A differentiable function \(J:\mathbb{R}^{n_\theta}\to\mathbb{R}\) is \(L_J\)-smooth in a domain $\Theta$ if \(\|\nabla J(\theta_1)-\nabla J(\theta_2)\|\le L_J\|\theta_1-\theta_2\|\) for all \(\theta_1,\theta_2\in\Theta\), and it is \(\mu\)-strongly convex if \(J(\theta_2)\ge J(\theta_1)+\nabla J(\theta_1)^\top(\theta_2-\theta_1)+\frac{\mu}{2}\|\theta_2-\theta_1\|^2\). Let $\Pi_{\mathcal{T}}(x)$ denote the standard Euclidean projection of $x$ onto a convex set $\mathcal{T}$. The interior of a set $\mathcal{T}$ is denoted by $\operatorname{int}(\mathcal{T})$.

\section{Personalized insulin delivery problem}
Consider the following dynamics for blood glucose,
\begin{subequations}
\label{eq:dynamics}
\begin{align}
\dot x&=f(x, u, d),\\
G&=g(x) + \nu,
\end{align}
\end{subequations}
where $x\in\reals^n$ is the system's state, $f:\reals^n\times\reals_+\times\reals^n\to\reals^n$ is the state transition map, $u\in\mathcal{U}\subseteq\reals_+$ is insulin, $d\in\mathcal{D}\subseteq\reals^n$ is the disturbance, $G\in\reals_+$ is the \ac{cgm} reading and $\nu\in\reals$ is \ac{cgm} noise. The disturbance $d$ captures any deviation from a baseline fasting glucose response including carbohydrate intake, and variability in behavior (e.g. physical activity) and metabolism (e.g. \ac{is}). 

The objective is minimizing the following glycemic risk
\begin{subequations}
\label{eq:risk_all}
\begin{equation}
\label{eq:risk}
    J_k = \textsc{lbgi}_k + r\,\textsc{hbgi}_k +  P_{\mathrm{treat},k}
\end{equation}
at each day $k$, where \ac{lbgi} and \ac{hbgi} are defined as
\begin{align}
    \textsc{lbgi}_k
    &= 68.31 \sum_{i=(k-1)N_s}^{kN_s-1}
    \mathbf{1}_{\{\textsc{bgi}_i < 90\}}\,\textsc{bgi}_i^2, \\
        \textsc{hbgi}_k
    &= 68.31 \sum_{i=(k-1)N_s}^{kN_s-1}
    \mathbf{1}_{\{\textsc{bgi}_i > 150\}}\,\textsc{bgi}_i^2. 
\end{align}
Here, $\mathbf{1}_{\{\cdot\}}\in\{0,1\}$ is the indicator function, and 
\begin{equation}
    \textsc{bgi}_i = \log \big( \mathrm{G}_i^{1.084} \big) - 5.381
\end{equation}
is computed from \ac{cgm} readings, which typically have a sampling period of $5$ minutes, yielding $N_s=288$ samples per day. The parameter $r$ allows for trade-off between hyperglycemia and hypoglycemia, with the baseline value $r=\frac{1}{3}$. Moreover, hypoglycemia treatments are penalized as
\begin{equation}
    P_{\mathrm{treat},k} = \left(\sum_{i=(k-1)N_s}^{kN_s-1} \mathrm{treat}_i\right)^{1.65},
\end{equation}
\end{subequations}
which discourages insulin overdoses that require administering rescue carbs in the future~\cite{moscoso2023reshaping}. 

% . Penalizing hypoglycemia treatments is crucial because they can mask \ac{lbgi}, leading to unsafe behavior if the hypoglycemia treatment pattern changes. This term therefore encourages minimally interventive closed-loop control

\textit{Analysis Surrogate:} We use a surrogate risk $\tilde{J}$, where step functions and event thresholds in~\eqref{eq:risk} are replaced by sigmoid approximations, while the exact clinical metric $J_k$ is retained for numerical evaluation. This is a standard relaxation in differentiable predictive control~\cite{cao2017extremum}. For notational brevity, we write $J_k$ throughout the analysis with the understanding that all gradient-based quantities refer to the surrogate $\tilde{J}_k$.

Our \ac{aid} system computes an insulin dose
\begin{equation}
\label{eq:ctrller}
    u=\kappa(x;\theta)= \gamma\beta u_\text{ss} + \gamma u_\text{\sc{bps}} + \gamma u_\text{\sc{umpc}}(\alpha,\beta),
\end{equation}
where the constant insulin $u_{ss}\in\reals_+$, the bolus insulin $u_\textsc{bps}\in\reals_+$ from the \ac{bps} controller, and the basal insulin $u_\textsc{umpc}\ge -u_{ss}$ from the \ac{umpc} controller are designed based on an average population (see~\cite{ shi2018adaptive}).

We exclusively tune $\theta = (\alpha, \beta) \in \mathbb{R}^2$ and leave the tuning of $\gamma$ to the existing method in~\cite{pryor2025pogo}. The parameter $\beta$ scales \ac{umpc}'s operating point, where the constant insulin $u_{op} = \gamma \beta u_{ss}$ drives the system toward a fasting glucose of $100\,\mathrm{mg/dL}$. The parameter $\alpha$ determines how conservatively \ac{umpc} handles insulin stacking. This distinction is clinically significant: recent \textit{in silico} and clinical studies show that switching to an ultra-rapid insulin formulation (e.g., Fiasp, AT247) yields minimal closed-loop improvement unless the controller is retuned to leverage the accelerated \ac{pkpd} profile~\cite{diaz2023maximizing}. The global scalar $\gamma$ cannot capture the altered temporal shape of the insulin action curve, since it scales the overall \ac{tdi} without distinguishing onset from tail. In contrast, tuning $\alpha$ allows the controller to safely reduce its conservativeness regarding insulin stacking, thus exploiting the faster clearance to curb postprandial peaks more aggressively while mitigating the risk of late-onset hypoglycemia. Concurrently, adjusting $\beta$ recalibrates the operating point for faster turnover. The following assumption underpins our convergence analysis.

\begin{assumption}[Contraction]
\label{ass:contraction}
The closed-loop system is contractive with rate $\lambda>0$. That is, there exists a Riemannian metric $M(x)$ such that, under any common input sequence $(\theta_t,d_t)\in\mathcal{T}\times\mathcal{D}$, any two trajectories $x(t)$ and $y(t)$, initialized in a forward-invariant set $\mathcal{C}$, satisfy
\[
d_M(x(t),y(t)) \le e^{-\lambda t} d_M(x(0),y(0))
\]
for all $t\ge 0$, where $d_M(\cdot,\cdot)$ denotes the distance induced by the metric $M(x)$.
\end{assumption}

Contraction does not imply disturbance rejection. Meals and physical activity can still cause large glucose excursions. Instead, contraction implies \emph{fading memory}, meaning that the long run response is dictated by the parameters and disturbance regardless of initial conditions. We use Assumption~\ref{ass:contraction} to show periodicity under periodic disturbance.

The intrinsic structure of \ac{t1d} physiology supports contractive behavior, as discussed in Remark~\ref{rem:structural}. Since modern \ac{umpc} and \ac{bps} controllers are designed to be stabilizing, it is plausible that the closed-loop system inherits contraction over physiologically relevant operating regimes. This assumption is further supported by the observed rapid decay of transient responses for the baseline model and some classes of metabolic variability~\cite{shi2018adaptive, pryor2025pogo,moscoso2023performance, villa2022model} because under mild technical assumptions, passivity imply contraction~\cite{kawano2024incremental}. 

\begin{remark}[Structural contractivity of T1D models]
\label{rem:structural}
In \ac{t1d} physiology, endogenous insulin secretion is absent. Thus, the open-loop dynamics decompose into a strict one-way cascade $\dot x_I = f_I(x_I,u,d)$ and $\dot x_G = f_G(x_G,x_I)$, where $x_I$ collects the insulin \ac{pkpd} and gut absorption states and $x_G$ the glucose compartments. The Jacobian is therefore block-triangular: $
\frac{\partial f}{\partial x}
=
\left[\begin{smallmatrix}
J_I(x_I) & 0\\
J_{GI}(x_G,x_I) & J_G(x_G,x_I)
\end{smallmatrix}\right]$.
One way to verify contractivity is that the Jacobian satisfies a Lyapunov inequality. There are physiological arguments supporting that each diagonal block is individually contracting;  for example, in Hovorka's model~\cite{hovorka2004nonlinear}, $J_I$ is contracting because \ac{pkpd} compartments exhibit strict mass dissipation, and $J_G$ because glucose is cleared by insulin-independent uptake and insulin-dependent utilization. When $J_I$ and $J_G$ exhibit contraction, by hierarchical contraction~\cite{bullo}, there exists a block-diagonal metric $M$ under which the full nonlinear system is contracting under the weighted norm induced by $M$, $\|x\|_M = \sqrt{x^\top M x}$, with rate $\lambda\approx\min(\lambda_I,\lambda_G)$. Similarly, Bergman's minimal model has a Jacobian $J = \left[\begin{smallmatrix} -(p_1+X) & -G & 0 \\ 0 & -p_2 & p_3 \\ 0 & 0 & -n \end{smallmatrix}\right]$, where $G$ is blood glucose, $X$ is insulin action, $p_1$ is glucose effectiveness, $p_2$ is the rate of insulin action decay, $p_3$ is \ac{is}, and $n$ is insulin clearance rate. Under $M=\mathrm{diag}(m_1,m_2,m_3)$, Sylvester's criterion on $MJ+J^\top M\prec 0$ requires $m_2/m_1 > G_{\max}^2/(4p_1 p_2)$ and $m_3/m_2 > p_3^2/(4p_2 n)$ which are physiologically admissible. The resulting contraction rate $\lambda \approx p_2 = 0.025$\,min$^{-1}$ yields a time constant of ${\sim}40$\,min, so transients wash out within $4$--$6$ hours, consistent with clinical observations~\cite{shi2018adaptive}.
\end{remark}

\begin{problem}
\label{prob:pas}
    Design an adaptation mechanism for $\theta_k\in\mathcal{T}$ to minimize the glycemic risk $J_k$ as $k\to\infty$ in the closed-loop architecture of Fig.~\ref{fig:block_diagram}.
\end{problem}
\begin{figure}[h]
\centering
\resizebox{\columnwidth}{!}{
\input{Block.tex}
}
\caption{Closed-loop fast control and slow adaptation.}
\label{fig:block_diagram}
\end{figure}
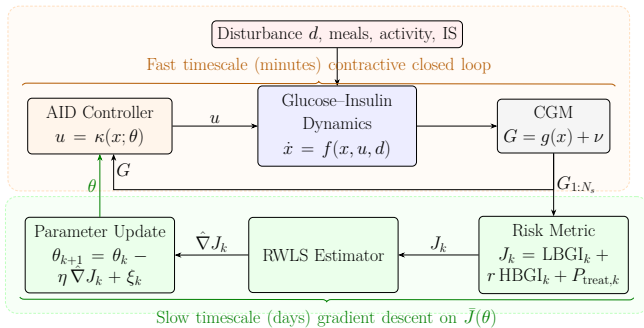
\section{Performance-based adaptation}

Our \ac{pas} adjusts $\theta_k\in\mathcal{T}\subseteq \reals^2$ by a projected real-time gradient descent
\begin{equation}
\label{eq:pas_update_main}
\theta_{k+1} = \Pi_{\mathcal{T}}\big( \theta_k - \eta \hat\nabla J_k + \xi_k \big),
\end{equation}
where $\Pi_{\mathcal{T}}(\cdot)$ is the Euclidean projection onto the admissible set $\mathcal{T}$, $\eta\in\reals_+$ is the learning rate, $\hat\nabla J_k\in\reals^2$ is an estimate of the surrogate risk gradient, and $\xi_k\in\reals^2$ is a bounded dither added for persistent exploration (e.g., periodic sign-flipping). Because $\mathcal{T}$ is convex, the projection is non-expansive ($\|\Pi_{\mathcal{T}}(x) - \Pi_{\mathcal{T}}(y)\| \le \|x-y\|$).

\subsection{Gradient estimation}
We employ a novel \ac{rwls} gradient estimator
\begin{subequations}
\label{eq:rwls}
    \begin{equation}
    \label{eq:rwls_main}
    \hat\nabla J_k
    =
    \big( \Delta \boldsymbol{\Theta}_k^\top W \Delta \boldsymbol{\Theta}_k+ \rho I \big)^{-1}
    \Delta \boldsymbol{\Theta}_k^\top W \Delta \mathbf{J}_k,
\end{equation}
where $W$ is a diagonal weighting matrix, $\rho$ is a regularization coefficient, and the data batches over a moving window of size $N$ are defined relative to the current operating point $\theta_k$,
\begin{align}
\Delta \mathbf{J}_k 
&= \big[ J_k - J_{k},\;  J_k - J_{k-1},\; \dots,\; J_k - J_{k-N} \big]^\top\!\!, \\
\Delta \boldsymbol{\Theta}_k 
&= \big[ \theta_k - \theta_{k},\; \theta_k - \theta_{k-1},\; \dots,\; \theta_k - \theta_{k-N} \big]^\top.
\end{align}
The entries $w_{i}, i\in\naturals_1^N,$ of the diagonal matrix $W$ are
\begin{equation}
\label{eq:weights}
   w_i=\frac{\hat w_i}{\sum_{j=1}^N \hat {w}_j}, \qquad \hat w_{i} = \frac{1}{\|\theta_k - \theta_{k-i}\|^3 + c^3},
\end{equation}
where $c > 0$ is the magnitude of dither. This weighting based on spatial proximity $\|\theta_k - \theta_{k-i}\|$ increases the contribution of samples $\{\theta_i,J_i\}$ that are closer to $\theta_k$, compensating for the higher-order terms in the Taylor expansion around the current point $\theta_k$. The parameter $\rho$ is an adaptive regularization,
\begin{equation}
\label{eq:rho}
    \rho = 0.1\left( \frac{\sum_{i=1}^N w_i \|\Delta\theta_i\|^4}{c^2} + c^2 \right),
\end{equation}
\end{subequations}
which tracks the weighted spatial variance of the data batch to make the step size adaptive to the estimation confidence; i.e., it prevents estimating a large gradient and taking large steps when measurements are distant from $\theta_k$~\cite{kashani2024robust, danielson2022extremum}.

\begin{lemma}[Bounded gradient estimation error]
\label{lem:estimator_robustness}
Let the steady-state surrogate risk $\bar J(\theta)$ be locally $L_J$-smooth within $\mathcal{T}$. Let the deterministic dither $\xi_k \in \reals^2$ be \ac{pe} such that $\Delta \boldsymbol{\Theta}_k^\top W \Delta \boldsymbol{\Theta}_k \succeq \sigma_0 I$ holds for some $\sigma_0 > 0$. Model the observed surrogate risk differences as $\Delta J_i = \nabla \bar J(\theta_k)^\top \Delta\theta_i + e_{\mathrm{curv},i} + e_{\mathrm{noise},i}$, where $|e_{\mathrm{curv},i}| \le \tfrac{L_J}{2}\|\Delta\theta_i\|^2$ and the residual day-to-day noise satisfies $|e_{\mathrm{noise},i}| \le \nu_{\max}$. Then, the estimation error $\delta_k := \hat\nabla J_k - \nabla\bar J(\theta_k)$ satisfies
\begin{equation}
\label{eq:delta_bound}
\|\delta_k\| \le \alpha \|\nabla \bar J(\theta_k)\| + \Delta_{\mathrm{est}},
\end{equation}
where $\Delta_{\mathrm{est}} := \frac{\sqrt{10}}{2c}\,\nu_{\max} + \frac{\sqrt{10}}{4}\, L_J\, c,$ and the relative bias factor satisfies $\alpha := \frac{\rho}{\sigma_0 + \rho} \in (0,1)$.
\end{lemma}

\begin{proof}
Substituting the observation model into~\eqref{eq:rwls_main} and letting $X=W^{1/2}\Delta\boldsymbol{\Theta}_k$ decomposes the error into bias and variance terms:
\begin{equation}
\begin{split}
    \delta_k =& \underbrace{-\rho(X^\top\! X+\rho I)^{-1}\nabla\bar J(\theta_k)}_{\text{bias}}
    + \\
    &\underbrace{(X^\top\! X+\rho I)^{-1}X^\top W^{1/2}(e_{\mathrm{curv}}+e_{\mathrm{noise}})}_{\text{variance}}.
\end{split}
\end{equation}

\emph{Bias term.} Because $X^\top X \succeq \sigma_0 I$, the spectral norm of the inverse satisfies $\|(X^\top X + \rho I)^{-1}\| \le (\sigma_0 + \rho)^{-1}$. Therefore, the bias magnitude is bounded by $\frac{\rho}{\sigma_0+\rho}\|\nabla\bar J\| = \alpha \|\nabla\bar J\|$.

\emph{Variance term.} Since $\|(X^\top\! X+\rho I)^{-1}X^\top\|\le 1/(2\sqrt\rho)$, the triangle inequality gives
\begin{equation}
\|\text{variance}\| \le \tfrac{1}{2\sqrt\rho}\bigl(\|W^{1/2}e_{\mathrm{noise}}\| + \|W^{1/2}e_{\mathrm{curv}}\|\bigr).
\end{equation}
Since $\sum w_i=1$, we have $\|W^{1/2}e_{\mathrm{noise}}\|^2 = \sum w_i e_{\mathrm{noise},i}^2 \le \nu_{\max}^2$, yielding $\|W^{1/2}e_{\mathrm{noise}}\|\le\nu_{\max}$. Because $\rho\ge 0.1c^2$ from~\eqref{eq:rho}, the noise amplification is bounded by $\frac{1}{2\sqrt{0.1c^2}}\nu_{\max}=\frac{\sqrt{10}}{2c}\,\nu_{\max}$.
By $L_J$-smoothness, $|e_{\mathrm{curv},i}|\le \frac{L_J}{2}\|\Delta\theta_i\|^2$. Let $V := \sum_{i=1}^N w_i\|\Delta\theta_i\|^4$ denote the weighted fourth spatial moment. Then $\|W^{1/2}e_{\mathrm{curv}}\|^2 \le \frac{L_J^2}{4}V$, yielding $\|W^{1/2}e_{\mathrm{curv}}\|\le \frac{L_J}{2}\sqrt{V}$. From~\eqref{eq:rho}, $\rho = 0.1(V/c^2 + c^2)$, so $\frac{1}{\sqrt\rho} = \frac{\sqrt{10}\,c}{\sqrt{V+c^4}}$. Thus:
\begin{equation}
\begin{split}
    \frac{1}{2\sqrt\rho}\|W^{1/2}e_{\mathrm{curv}}\| & \le \frac{\sqrt{10}\,c}{2\sqrt{V+c^4}} \cdot \frac{L_J}{2}\sqrt{V} \\
    & = \frac{\sqrt{10}}{4} L_J c \sqrt{\frac{V}{V+c^4}} \le \frac{\sqrt{10}}{4} L_J c.
\end{split}
\end{equation}
Summing the noise and curvature bounds yields $\Delta_{\mathrm{est}}$.
\end{proof}

\begin{remark}
\label{rem:weight_design}
By formulating the adaptive regularization~\eqref{eq:rho} using the weighted spatial variance of the data batch, the estimator acts as a dynamically scaled trust region. This yields an error bound~\eqref{eq:delta_bound} that recovers the fundamental trade-off: a larger exploration dither $c$ attenuates noise ($\propto 1/c$) but amplifies curvature bias ($\propto c$).
\end{remark}

\subsection{From contraction to steady-state optimization}
\label{sec:periodic}
In this section, we show that \ac{pas} drives the parameter to a neighborhood of the optimum under some regularity assumptions, while preserving the closed-loop \ac{iss} property. The \ac{iss} property is important because the gradient descent update~\eqref{eq:pas_update_main} acts as an integral action, which can destabilize the system. We establish this result in two steps. First, under a daily routine with periodic disturbances and exact gradient estimation, we show that the system converges to the optimum. Then, in a more realistic setting with random disturbances, we show that the system converges to a neighborhood of the optimum. 

\begin{lemma}[Theorem~$3.15$ from~\cite{bullo}]
\label{lemma:periodic}
Let Assumption~\ref{ass:contraction} hold. Let $\theta_k=\bar\theta\in \mathcal{T}$ be constant. Let $d_t$ be periodic with period $T$. Then, there exists a unique $T$-periodic solution $x^\star_t(\bar\theta)\in\mathcal{C}, \forall t$ to the closed-loop system such that for any initial condition $x_0\in\mathcal{C}$, the corresponding trajectory satisfies
\begin{equation}
d_M(x_t,x^\star_t(\bar\theta))\le e^{-\lambda(t-t_0)}d_M(x_0,x^\star_{0}(\bar\theta)),
\end{equation}
for all $t\ge t_0$, where $\lambda>0$ is the contraction rate.
\end{lemma}

Lemma~\ref{lemma:periodic} states that when daily disturbances repeat, so does the glucose response, even though individual meals and activity bouts still cause large transient excursions. The key requirement is periodicity of the disturbance pattern, not its magnitude. Clinicians implicitly rely on this entrainment when designing periodic \ac{mpc} for \ac{aid}~\cite{gondhalekar2016periodic}. However, updating $\theta_k$ changes the periodic orbit $x^\star_t(\bar\theta)$. The following theorem shows that there is a limit on how fast $\theta_k$ can be updated. 
% \begin{proposition}
% \label{prop:cost-converge}
% Let the assumptions of Lemma~\ref{lemma:periodic} hold with contraction rate $\lambda>0$ and period $T$ of one day. Let the hypoglycemia treatment $P_{\mathrm{treat},k}(g)$ be a Lipschitz function of \ac{cgm} measurement $G$. Let \ac{cgm} data be noise-free $\nu=0$. Then the risk metric $J_k$ converges exponentially to a constant $\bar J$, i.e., $|J_k - \bar J|\le ce^{-\lambda k},$ for some constant $c>0$.
% \end{proposition}
% \begin{proof}
% By Lemma~\ref{lemma:periodic}, the glucose trajectory $g(x_t)$ of the system~\eqref{eq:dynamics} converges exponentially to a unique daily periodic orbit $g^\star_t = g(x^\star_t)$ with rate $\lambda>0$.  
% Since $\text{\sc{bgi}}_t$ and $P_{\mathrm{treat},k}$ are Lipschitz functions of glucose, they also converge exponentially to their respective periodic trajectories. Similarly, finite sums and linear combinations in $\text{\sc{lbgi}}_k$ and $\text{\sc{hbgi}}_k$ are Lipschitz continuous and thus preserve exponential convergence. Hence, $J_k$ converges exponentially to a periodic trajectory $\bar J_k = \sum_{t=(k-1)T}^{kT-1} \phi(g^\star_{t})$ for some $\phi$, i.e., there exists $c>0$ such that
% \begin{subequations}
% \begin{equation}
% |J_k - \bar J_k| \le c e^{-\lambda k}.
% \end{equation}
% Finally, since $g^\star_t$ is $T$-periodic, 
% \begin{equation}
%     \bar J_{k+1} = \sum_{t=kT}^{(k+1)T-1} \phi(g^\star_{t})
%     = \sum_{t=(k-1)T}^{kT-1} \phi(g^\star_{t+T})
%     = \bar J_k
% \end{equation} 
% \end{subequations}
% which shows that $\bar J_k=\bar J$ is constant.
% \end{proof}
\begin{theorem}[Convergence to optimum]
\label{thm:joint_stability}
Let Assumption~\ref{ass:contraction} hold. Let the disturbance \(d_t\) be \(T\)-periodic. Assume:
\begin{enumerate}
\item The periodic orbit $x_t^\star(\theta)$ is Lipschitz in $\Theta$, implying $\chi^2_k
\le
L_{x^\star}^2\|\theta_1-\theta_2\|^2$ where 
\begin{equation}
\label{eq:tilde_def}
\chi^2_k
:=
\frac1T\int_{(k-1)T}^{kT}d_M(x_t,x_t^\star(\theta_k))^2dt,
\end{equation}
\item The steady-state risk \(\bar J(\theta)\) is \(\mu\)-strongly convex on \(\Theta\), admits a unique minimizer \(\theta^\star\in\Theta\), and is \(L_J\)-smooth.
\item The adaptation law~\eqref{eq:pas_update_main} is dither-free and uses the exact gradient of the steady-state risk \(\hat\nabla J_k=\nabla \bar J(\theta_k)\).
\end{enumerate}
Then, the closed-loop system converges exponentially to the periodic orbit \((x_t^\star(\theta^\star),\theta^\star)\) if, for any $p>0$, $\eta$ satisfies
\begin{equation}
\label{eq:eta_final}
0<\eta<
\frac{2p\mu}{
L_J^2\left(\left(1+\frac{2}{e^{2\lambda T}-1}\right)L_{x^\star}^2+p\right)
}.
\end{equation}
\end{theorem}

\begin{proof}
Lemma~\ref{lemma:periodic} implies $
\chi^2_{k+1}
\le
e^{-2\lambda T}\chi^2_k$ when $\theta_{k+1}=\theta_k$. However, the parameter update \(\theta_{k+1}=\theta_k-\eta\nabla \bar J(\theta_k)\) changes the reference orbit. We expand
\begin{equation}
    d_M(x_t,x_t^\star(\theta_{k+1}))=d_M(x_t,x_t^\star(\theta_{k})+x_t^\star(\theta_{k})-x_t^\star(\theta_{k+1})),
\end{equation} 
and use Lipschitz continuity of \(x_t^\star(\theta)\) and Young's inequality $\|a+b\|^2 \le (1+\epsilon)\|a\|^2 + (1+\frac{1}{\epsilon})\|b\|^2$ with \(\epsilon>0\) to get
\begin{equation}
\label{eq:x_bound_young}
\chi^2_{k+1}
\le
(1+\epsilon)e^{-2\lambda T}\chi^2_k
+
\left(1+\frac{1}{\epsilon}\right)L_{x^\star}^2\|\theta_{k+1}-\theta_k\|^2.
\end{equation}
Applying the update~\eqref{eq:pas_update_main} on \eqref{eq:x_bound_young}, we have
\begin{equation}
\label{eq:x_bound_final}
\chi^2_{k+1}
\le
(1+\epsilon)e^{-2\lambda T}\chi^2_k
+
\left(1+\frac{1}{\epsilon}\right)L_{x^\star}^2\eta^2L_J^2\|\tilde{\theta}_k\|^2,
\end{equation}
where $\tilde{\theta}_k:=\theta_k-\theta^\star$ and nonexpansiveness of $\Pi_\mathcal{T}$ and \(L_J\)-smoothness of $J$ is leveraged. Further, by $\mu$-strong convexity and \(L_J\)-smoothness, for some \(0<\eta<\frac{2\mu}{L_J^2}\), we have
\begin{equation}
\label{eq:theta_contraction}
\|\tilde{\theta}_{k+1}\|^2
\le
\bigl(1-2\eta\mu+\eta^2L_J^2\bigr)\|\tilde{\theta}_k\|^2
=:q\|\tilde{\theta}_k\|^2,
\end{equation}
where \(q\in(0,1)\). Now consider the Lyapunov function
\begin{equation}
\label{eq:lyapunov}
V_k=\chi^2_k+p\|\tilde{\theta}_k\|^2.
\end{equation}
Combining \eqref{eq:x_bound_final} and \eqref{eq:theta_contraction} gives
\begin{equation}
\label{eq:V_bound}
V_{k+1}
\le
(1+\epsilon)e^{-2\lambda T}\chi^2_k
+
\left[
\left(1+\frac{1}{\epsilon}\right)L_{x^\star}^2\eta^2L_J^2
+
pq
\right]\|\tilde{\theta}_k\|^2.
\end{equation}
Let us define the following rates
\begin{equation}
\label{eq:sigma_def}
\sigma_x := (1+\epsilon)e^{-2\lambda T},
\>
\sigma_\theta := q+\frac{1}{p}\left(1+\frac{1}{\epsilon}\right)L_{x^\star}^2\eta^2L_J^2.
\end{equation}
Then from \eqref{eq:V_bound} we have $V_{k+1}\le \sigma_x \chi^2_k+\sigma_\theta\,p\|\tilde{\theta}_k\|^2.$
Hence, we obtain contraction $V_{k+1}\le \sigma V_k$ with $\sigma:=\max\{\sigma_x,\sigma_\theta\}\in(0,1)$ when \(\sigma_x<1\) and \(\sigma_\theta<1\). To ensure \(\sigma_x<1\), we choose
$\epsilon = \frac{1}{2}(e^{2\lambda T}-1)$. Likewise, \(\sigma_\theta<1\) holds when $\left(1+\frac{1}{\epsilon}\right)L_{x^\star}^2\eta^2L_J^2
<
p(1-q).$
Since $1-q=2\eta\mu-\eta^2L_J^2,$ the \(\sigma_x<1\) holds when learning rate $\eta$ satisfies~\eqref{eq:eta_final}. Therefore, \(V_k\) is decreasing which supports exponential convergence \((x_t, \theta_k)\to(x_t^\star(\theta^\star), \theta^\star)\).
\end{proof}

Theorem~\ref{thm:joint_stability} derives a bound~\eqref{eq:eta_final} on the learning rate $\eta$ that preserves the closed-loop \ac{iss} property and ensures convergence to the optimum under periodic disturbances and exact gradient estimation. The following corollary promotes closed-loop \ac{iss} property and convergence to a neighborhood of optimum when the estimator~\eqref{eq:rwls_main} is used and the system is subject to noise and changes in the disturbance pattern.

\begin{corollary}[Convergence under estimated gradient]
\label{cor:joint_stability_transient}
Let the assumptions of Theorem~\ref{thm:joint_stability} hold. Assume the update law~\eqref{eq:pas_update_main} uses the \ac{rwls} estimator with a bounded dither $\xi_k$ satisfying the \ac{pe} condition of Lemma~\ref{lem:estimator_robustness}. Assume the gradient evaluated at the state trajectory satisfies a Lipschitz bound with respect to the steady-state orbit:
\begin{equation}
\label{eq:gradient_mismatch}
\|\nabla_\theta \tilde{J}_k(x,\theta_k)-\nabla_\theta \bar J(x^\star(\theta_k), \theta_k)\| \le L_g |\chi_k|
\end{equation}
where, with abuse of notation, $x$ and $x^\star$ are daily trajectories.
Provided the \ac{pe} guarantees, the relative estimator bias satisfies $\alpha < \frac{\mu}{\sqrt{6} L_J}$, there exists a sufficiently small learning rate \(\eta>0\) such that the closed-loop system is \ac{iss}, converging to a neighborhood of \((x_t^\star(\theta^\star),\theta^\star)\) within $\reals^n\times\mathcal{T}$.
\end{corollary}

\begin{proof}
Let $\delta_k := \hat{\nabla} \tilde{J}_k - \nabla \bar J(\theta_k)$. Applying the triangle inequality to~\eqref{eq:gradient_mismatch} and Lemma~\ref{lem:estimator_robustness} yields:
\begin{equation}
\|\delta_k\| \le \alpha\|\nabla \bar J(\theta_k)\| + L_g |\chi_k| + \Delta_{\mathrm{est}}.
\end{equation}
Because $\theta^\star \in \operatorname{int}(\mathcal{T})$ and $\bar J$ is locally $L_J$-smooth, $\nabla \bar J(\theta^\star) = 0$ and thus $\|\nabla \bar J(\theta_k)\| \le L_J \|\tilde{\theta}_k\|$. Using the inequality $(a+b+c)^2 \le 3a^2 + 3b^2 + 3c^2$, the mismatch is bounded by:
\begin{equation}
\label{eq:delta_k_bound}
\|\delta_k\|^2 \le 3\alpha^2 L_J^2 \|\tilde{\theta}_k\|^2 + 3L_g^2 \chi^2_k + 3\Delta_{\mathrm{est}}^2.
\end{equation}

By the non-expansive property of the projection operator $\Pi_{\mathcal{T}}$, the distance to the optimum $\theta^\star \in \mathcal{T}$ satisfies $
\|\tilde{\theta}_{k+1}\|^2 \le \|\tilde{\theta}_k - \eta \nabla \bar J(\theta_k) - \eta \delta_k + \xi_k\|^2$.
Applying Young's inequality $\|a+b\|^2 \le (1+\epsilon)\|a\|^2 + (1+\frac{1}{\epsilon})\|b\|^2$ with $a = \tilde{\theta}_k - \eta \nabla \bar J(\theta_k)$ and $b = -\eta \delta_k + \xi_k$ yields:
\begin{align}
\|\tilde{\theta}_{k+1}\|^2 
\le & (1+\epsilon)\|\tilde{\theta}_k - \eta \nabla \bar J(\theta_k)\|^2 \notag \\
&
+ \left(1+\frac{1}{\epsilon}\right)\bigl( 2\eta^2\|\delta_k\|^2 + 2\|\xi_k\|^2 \bigr)
\end{align}
for any $\epsilon > 0$. By local $\mu$-strong convexity and $L_J$-smoothness, $\|\tilde{\theta}_k - \eta \nabla \bar J(\theta_k)\|^2 \le (1 - 2\eta\mu + \eta^2 L_J^2)\|\tilde{\theta}_k\|^2$. We select $\epsilon = \eta\mu$. Substituting~\eqref{eq:delta_k_bound} into the parameter bound, the coefficient of $\|\tilde{\theta}_k\|^2$ becomes:
$\rho_\theta(\eta) = (1+\eta\mu)\left(1 - 2\eta\mu + \eta^2 L_J^2\right) + \left(1+\frac{1}{\eta\mu}\right) 6\eta^2 \alpha^2 L_J^2.$
Expanding this algebraically:
$\rho_\theta(\eta) = 1 - \eta\mu - 2\eta^2\mu^2 + \eta^2 L_J^2 + \eta^3\mu L_J^2 + 6\eta^2 \alpha^2 L_J^2 + \frac{6\eta\alpha^2 L_J^2}{\mu}$.
For the parameter tracking error to strictly contract ($\rho_\theta < 1$), we group the linear terms in $\eta$:
\begin{equation}
\rho_\theta(\eta) = 1 - \eta\left( \mu - \frac{6\alpha^2 L_J^2}{\mu} \right) + \mathcal{O}(\eta^2).
\end{equation}
For this to be strictly less than $1$ for sufficiently small $\eta > 0$, the linear coefficient of $\eta$ must be strictly negative. This rigidly requires $\mu - \frac{6\alpha^2 L_J^2}{\mu} > 0$, which mathematically simplifies to $\alpha < \frac{\mu}{\sqrt{6} L_J}$. Provided the dither provides sufficient PE to satisfy this bound, we can select $\eta$ small enough such that the strictly negative linear term dominates the quadratic terms, ensuring there exists $\bar{q} \in (0,1)$ such that $\|\tilde{\theta}_{k+1}\|^2 \le \bar q\|\tilde{\theta}_k\|^2 + \bar c_x\eta^2\chi^2_k + \bar c_{\Delta}\eta^2\Delta_{\mathrm{est}}^2 + c_\xi\|\xi_k\|^2$. 

Concurrently, by non-expansiveness of $\Pi_{\mathcal{T}}$, the step size satisfies $\|\theta_{k+1}-\theta_k\| \le \eta\|\nabla\bar J(\theta_k)\| + \eta\|\delta_k\| + \|\xi_k\|$. Expanding the state bound~\eqref{eq:x_bound_final} with this projected-step bound yields a similar affine bound in $\chi^2_k$, $\|\tilde{\theta}_k\|^2$, $\Delta_{\mathrm{est}}^2$, and $\|\xi_k\|^2$. Combining these into the Lyapunov function $V_k=\chi^2_k+p\|\tilde{\theta}_k\|^2$, there exists $\sigma \in (0,1)$ such that $V_{k+1} \le \sigma V_k + \mathcal{O}(\eta^2 \Delta_{\mathrm{est}}^2 + \|\xi_k\|^2)$. Thus the closed-loop system is \ac{iss}. The parameters and physiological states converge to an ultimate bounded neighborhood dictated by the persistent exploration dither and residual metabolic noise, localized within $\mathcal{T}$ by the projection operator.
\end{proof}
\section{\textit{In silico} Trial}
We use the latest version of the UVA/Padova simulator to evaluate our \ac{pas} algorithm under meal disturbances, diurnal variability in insulin sensitivity, and dawn phenomenon. The simulator provides a cohort of $100$ \textit{in silico} adult subjects spanning a range of human metabolic parameters, whose key demographic characteristics are summarized in Table~\ref{tab:demo}. 

\begin{table}[h]
\centering
\caption{Key demographic of the \textit{in silico} adult cohort}
\begin{tabular}{lcccc}
\hline
Parameter & Units & Mean (\sc{sd}) & Min & Max \\
\hline
Body weight & kg & 75.2 (12.1) & 52.6 & 108.4 \\
\ac{tdi} & U & 40.4 (12.98) & 21.2 & 104.7 \\
\ac{tdi}/kg & U/kg & 0.54 (0.15) & 0.33 & 1.22 \\
Fasting plasma glucose & mg/dl & 119.1 (7.1) & 102.1 & 134.5 \\
Duration of diabetes & years & 23.9 (11.4) & 3 & 52 \\
\hline
\label{tab:demo}
\end{tabular}
\end{table}

We simulated \(100\) \textit{in silico} adults under two scenarios: 1) without variability, with constant meals and metabolic parameters, and 2) with meal and metabolic variability. Our \ac{pas} uses a learning rate of \(\eta=1\) and dither amplitude \(c=0.1\), where the dither direction flips every day for \(\beta\) and every two days for \(\alpha\).

In the no-variability scenario, the meal carbohydrate content is $105$\,g for breakfast, lunch, and dinner at $8$, $13$, and $19$ hours, respectively. In the variable scenario, the meal carbohydrate content is $84 \pm 20$\,g with the same average timing and a timing variability of $\pm 20$ minutes. Three snacks are given randomly with $24 \pm 14$\,g at $10$, $16$, and $21$ hours, each with a timing variability of $\pm 20$ minutes. Hypoglycemia rescue carbohydrates are administered when glucose falls below $70$\,mg/dL followed by a $15$ minutes waiting period before delivering additional treatment. We use insulin Aspart and incorporate both intra-day and inter-day variability through variations in insulin sensitivity.

For the no-variability scenario, Fig.~\ref{fig:Jnovar} shows that the average risk decreases under adaptation compared to no adaptation, and that the parameters converge over time. In Fig.~\ref{fig:Jnovar}, the dashed yellow curve corresponds to an existing \(\gamma\)-adaptation method for scaling \ac{tdi}~\cite{pryor2025pogo}. The split blue curve shows that adapting \((\alpha,\beta)\) can achieve lower risk than adapting \(\gamma\) alone. The unevenly dashed purple curve represents a case in which \(\gamma\) is adapted by the existing method for one week, after which our \ac{pas} takes over and adapts \((\alpha,\beta)\), leading to convergence to a lower overall risk. Similarly, Fig.~\ref{fig:Jvar} shows the results for the case with variable meals and metabolism. Fig.~\ref{fig:param} shows parameter convergence for subjects $1-5$. Table~\ref{tab:stat} shows \ac{tir} ($G\in[70,180]$), \ac{tar} ($g>180$), and \ac{tbr} ($g<70$) with their corresponding \ac{sd} for the two scenarios.

\begin{figure}
    \centering
    \includegraphics[width=\linewidth]{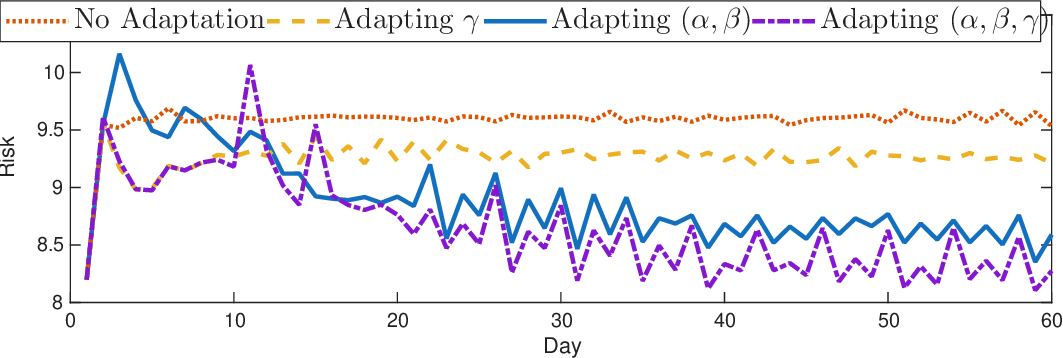}
    \caption{Mean of risk $J$ for $100$ subjects \textbf{without} variability.}
    \label{fig:Jnovar}
\end{figure}

\begin{figure}
    \centering
    \includegraphics[width=\linewidth]{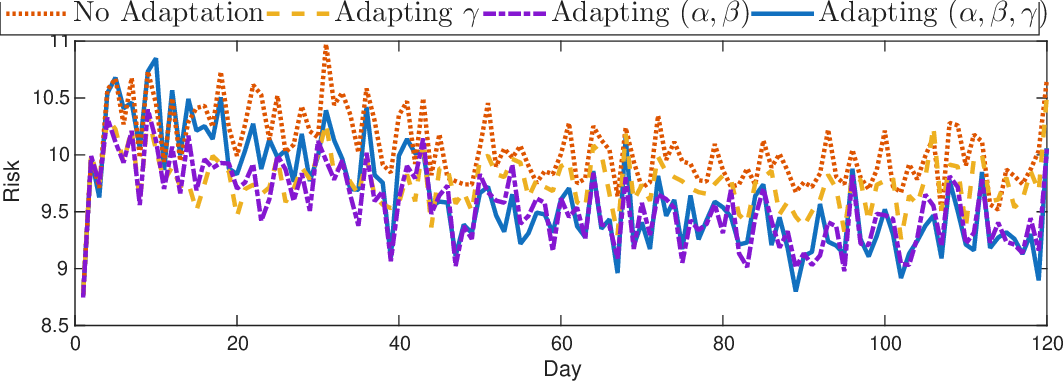}
    \caption{Mean of risk $J$ for $100$ subjects \textbf{with} variability.}
    \label{fig:Jvar}
\end{figure}

\begin{figure}[t]
    \centering

    \begin{subfigure}{0.49\linewidth}
        \centering
        \includegraphics[width=\linewidth]{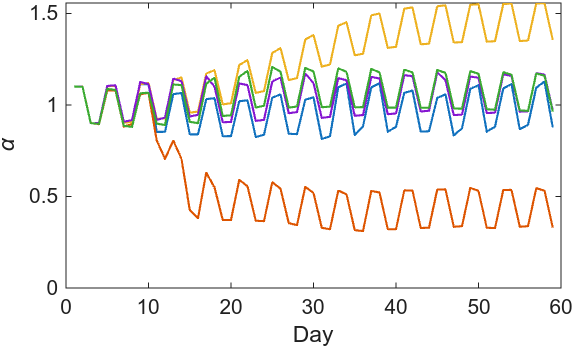}
        \caption{No variability}
    \end{subfigure}
    \hfill
    \begin{subfigure}{0.49\linewidth}
        \centering
        \includegraphics[width=\linewidth]{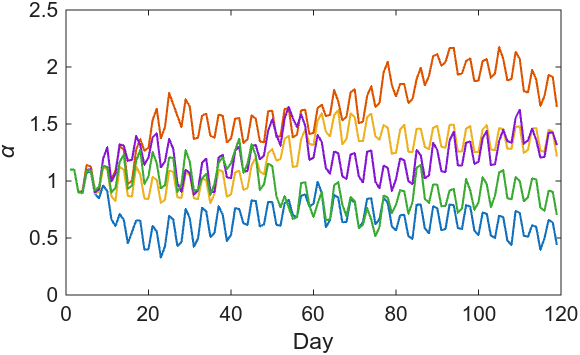}
        \caption{With variability}
    \end{subfigure}

    \begin{subfigure}{0.49\linewidth}
        \centering
        \includegraphics[width=\linewidth]{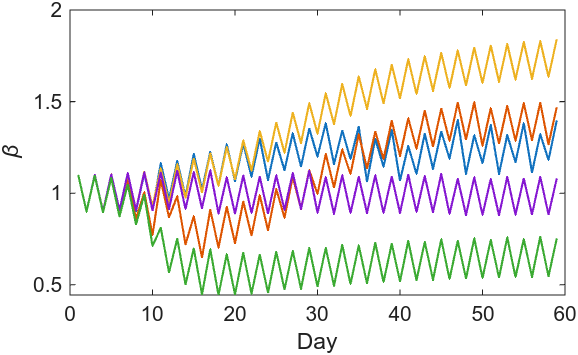}
        \caption{No variability}
    \end{subfigure}
    \hfill
    \begin{subfigure}{0.49\linewidth}
        \centering
        \includegraphics[width=\linewidth]{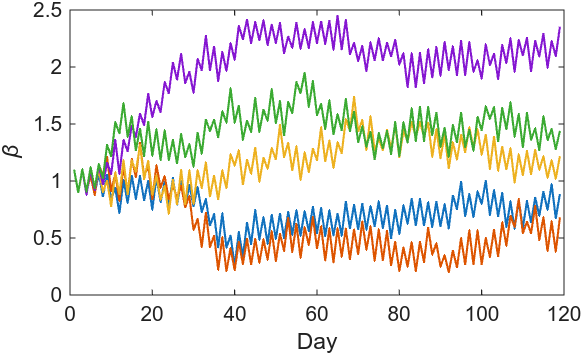}
        \caption{With variability}
    \end{subfigure}
    
    \caption{Adaptation of $\theta=(\alpha,\beta)$ for subjects $1-5$ under scenarios with and without meal and metabolic variability.}
    \label{fig:param}
\end{figure}

\input{table}

\section{Conclusion}
We presented a real-time, data-driven method for personalizing controller parameters in \ac{aid} systems. The multi-parameter framework supports integration of ultra-rapid insulin formulations, where global \ac{tdi} scaling alone cannot account for the altered \ac{pkpd} temporal profile. The adaptation law achieves input-to-state stability under disturbances.

% To overcome the brittle nature of standard gradient estimators in noisy biomedical applications, we introduced an \ac{rwls} estimator with spatial proximity weighting and adaptive regularization. We proved that under deterministic \ac{pe}, this estimator inherently bounds Taylor truncation errors and caps noise amplification. By projecting the parameter updates onto a compact admissible domain, we guaranteed robust \ac{iss} property even under persistent bounded metabolic disturbances. \textit{In silico} trials confirmed that our projected gradient descent algorithm dynamically customizes controller aggressiveness and steady-state targeting, achieving significant clinical improvements in Time-in-Range.

% Furthermore, our multi-parameter framework inherently supports the safe integration of ultra-rapid insulin formulations. As adjusting global insulin scalar factors cannot fully account for PK/PD phase shifts, dynamically adapting the aggressiveness and steady-state targeting parameters ($\alpha$ and $\beta$) empowers the controller to safely exploit the faster clearance of modern insulins. Because faster insulin clears the system more rapidly, the controller can respond more aggressively to postprandial hyperglycemia with a significantly reduced risk of delayed insulin stacking---a phenomenon that traditionally endangers patients during unpredictable follow-up physical exertion.

\bibliographystyle{IEEEtran}
\bibliography{references}

\end{document}

%% file: Packages.tex
\usepackage{amsmath,amssymb,bm} % for math symbols and notation
\usepackage{acronym}
\usepackage{color}
\usepackage{hyperref}
\usepackage{graphicx}		
\usepackage{mathtools}
\usepackage[normalem]{ulem}
\usepackage{comment}
\usepackage{algorithm,algorithmic}
\usepackage{cite}
\usepackage{xcolor}
\usepackage{subcaption}
\usepackage{tikz}
\usetikzlibrary{positioning, arrows.meta, calc, fit, backgrounds, decorations.pathreplacing}
\usepackage{multirow}
\usepackage{array}

\acrodef{cgm}[{\sc cgm}]{continuous glucose monitoring}
\acrodef{ap}[AP]{artificial pancreas}
\acrodef{aid}[{\sc aid}]{automated insulin delivery}
\acrodef{t1d}[{\sc t1d}]{type 1 diabetes}
\acrodef{mpc}[{\sc mpc}]{model-predictive control}
\acrodef{ls}[{\sc ls}]{least squares}
\acrodef{esc}[{\sc esc}]{extremum-seeking control}
\acrodef{rwls}[{\sc rwls}]{regularized weighted least squares}
\acrodef{gpi}[{\sc gpi}]{glucose penalty index}
\acrodef{tir}[{\sc tir}]{time-in-range}
\acrodef{titr}[{\sc titr}]{time-in-tight-range}
\acrodef{tar}[{\sc tar}]{time-above-range}
\acrodef{tbr}[{\sc tbr}]{time-below-range}
\acrodef{tdi}[{\sc tdi}]{total daily insulin}
\acrodef{bgi}[{\sc bgi}]{blood glucose index}
\acrodef{lbgi}[{\sc lbgi}]{low blood glucose index}
\acrodef{hbgi}[{\sc hbgi}]{high blood glucose index}
\acrodef{pas}[{\sc pas}]{performance-based adaptation system}
\acrodef{iob}[{\sc iob}]{insulin on board}
\acrodef{umpc}[{\sc umpc}]{University of Virginia \ac{mpc}}
\acrodef{bps}[{\sc bps}]{bolus priming system}
\acrodef{iss}[{\sc iss}]{input-to-state stable}
\acrodef{sd}[{\sc sd}]{standard deviation}
\acrodef{is}[{\sc is}]{insulin sensitivity}
\acrodef{pkpd}[{\sc pk/pd}]{pharmacokinetics/pharmacodynamics}
\acrodef{pe}[{\sc pe}]{persistently exciting}

%% file: CustomCommands.tex
\newcommand{\reals}{\mathbb{R}}
\newcommand{\naturals}{\mathbb{N}}

\newtheorem{theorem}{Theorem}
\newtheorem{corollary}{Corollary}
\newtheorem{lemma}{Lemma}
\newtheorem{remark}{Remark}

\newtheorem{problem}{Problem}

\newtheorem{assumption}{Assumption}

%% file: Block.tex
\begin{tikzpicture}[
    >={Stealth[length=7pt]},
    font=\huge,
    ctrl/.style={
        draw, thick, rounded corners=5pt,
        fill=orange!8, align=center,
        minimum width=5.6cm, minimum height=2.2cm,
        text width=5.2cm,
        inner sep=8pt
    },
    plant/.style={
        draw, thick, rounded corners=5pt,
        fill=blue!7, align=center,
        minimum width=6.1cm, minimum height=2.5cm,
        text width=5.7cm,
        inner sep=9pt
    },
    sensor/.style={
        draw, thick, rounded corners=5pt,
        fill=gray!8, align=center,
        minimum width=4.4cm, minimum height=1.9cm,
        text width=4.0cm,
        inner sep=7pt
    },
    adapt/.style={
        draw, thick, rounded corners=5pt,
        fill=green!8, align=center,
        minimum width=5.8cm, minimum height=3cm,
        text width=5.4cm,
        inner sep=8pt
    },
    dist/.style={
        draw, thick, rounded corners=5pt,
        fill=red!6, align=center,
        minimum width=9.0cm, minimum height=1.0cm,
        inner sep=6pt
    },
    arr/.style={->, thick},
    lbl/.style={font=\huge},
]

\node[ctrl] (aid) {%
    AID Controller\\[4pt]
    $u=\kappa(x;\theta)$%
};

\node[plant, right=3.4cm of aid] (system) {%
    Glucose--Insulin\\[2pt]
    Dynamics\\[4pt]
    $\dot x=f(x,u,d)$%
};

\node[sensor, right=3.2cm of system] (cgm) {%
    CGM\\[4pt]
    $G=g(x)+\nu$%
};

\node[dist, above=1.5cm of system] (disturbance) {%
    Disturbance $d$, meals, activity, IS%
};

\draw[arr] (aid) -- node[lbl, above] {$u$} (system);
\draw[arr] (disturbance) -- (system);
\draw[arr] (system) -- (cgm);

\node[adapt, below=2.5cm of aid] (update) {%
    Parameter Update\\[4pt]
    $\theta_{k+1}=\theta_k-\eta\,\hat{\nabla}J_k+\xi_k$%
};

\node[adapt, below=2.5cm of cgm] (risk) {%
    Risk Metric\\[4pt]
    $J_k=\mathrm{LBGI}_k+r\,\mathrm{HBGI}_k+P_{\mathrm{treat},k}$%
};

\node[adapt, left=3.2cm of risk] (rwls) {%
    RWLS Estimator\\[2pt]
};

\draw[arr] (risk.west) -- node[lbl, above] {$J_k$} (rwls.east);
\draw[arr] (rwls.west) -- node[lbl, above] {$\hat{\nabla}J_k$} (update.east);

\draw[arr, draw=green!50!black]
    (update.north) -- node[lbl, left, text=green!50!black] {$\theta$} (aid.south);

\draw[arr] (cgm.south) -- node[lbl, right] {$G_{1{:}N_s}$} (risk.north);

\coordinate (gdrop) at ([yshift=-1.45cm]cgm.south);
\coordinate (gin) at ([xshift=0.55cm]aid.south);
\draw[thick] (cgm.south) -- (gdrop);
\draw[arr] (gdrop) -| (gin);
\node[lbl] at ([xshift=0.4cm,yshift=-0.6cm]gin) {$G$};

\draw[decorate, decoration={brace, amplitude=5pt}, thick, draw=orange!70!black]
    ([yshift=0.5cm, xshift=-0.2cm]aid.north west) --
    ([yshift=0.5cm, xshift=0.2cm]cgm.north east)
    node[midway, above=9pt, font=\huge, text=orange!70!black] {%
        Fast timescale (minutes) contractive closed loop%
    };

\draw[decorate, decoration={brace, amplitude=5pt, mirror}, thick, draw=green!50!black]
    ([yshift=-0.15cm, xshift=-0.2cm]update.south west) --
    ([yshift=-0.15cm, xshift=0.2cm]risk.south east)
    node[midway, below=9pt, font=\huge, text=green!50!black] {%
        Slow timescale (days) gradient descent on $\bar J(\theta)$%
    };

\begin{scope}[on background layer]
    \node[
        fit=(aid)(system)(cgm)(disturbance),
        inner sep=22pt,
        yshift=-0.2cm,
        draw=orange!25,
        rounded corners=13pt,
        fill=orange!3,
        densely dashed
    ] {};
    \node[
        fit=(update)(rwls)(risk),
        inner sep=22pt,
        draw=green!25,
        rounded corners=13pt,
        fill=green!3,
        densely dashed
    ] {};
\end{scope}

\end{tikzpicture}

%% file: table.tex
\begin{table}[h]
\caption{Glycemic outcomes for two scenarios.}
\label{tab:stat}
\centering
\setlength{\tabcolsep}{4pt}
\renewcommand{\arraystretch}{1.1}
\begin{tabular}{|c|c|c|c|c|c|}
\hline
\rotatebox{90}{Vary} &
\rotatebox{90}{Week} &
{Adapt} &
\ac{tir}(\ac{sd}) &
\ac{tar}(\ac{sd}) &
\ac{tbr}(\ac{sd}) \\
\hline

\multirow{2}{*}{No}
& 4
& \begin{tabular}[c]{@{}c@{}}
No \\
$\gamma$ \\
$(\alpha,\beta)$ \\
$\gamma,\;(\alpha,\beta)$
\end{tabular}
& \begin{tabular}[c]{@{}c@{}}
67.52 (8.73) \\
67.54 (9.23) \\
68.20 (10.27) \\
67.43 (10.11)
\end{tabular}
& \begin{tabular}[c]{@{}c@{}}
31.23 (8.05) \\
31.41 (8.47) \\
30.96 (9.79) \\
31.89 (9.70)
\end{tabular}
& \begin{tabular}[c]{@{}c@{}}
1.25 (1.59) \\
1.04 (0.95) \\
0.84 (1.39) \\
0.68 (1.02)
\end{tabular}
\\ \cline{2-6}

& 8
& \begin{tabular}[c]{@{}c@{}}
No \\
$\gamma$ \\
$(\alpha,\beta)$ \\
$\gamma,\;(\alpha,\beta)$
\end{tabular}
& \begin{tabular}[c]{@{}c@{}}
67.50 (8.72) \\
67.56 (9.24) \\
69.17 (10.53) \\
68.80 (10.85)
\end{tabular}
& \begin{tabular}[c]{@{}c@{}}
31.25 (8.04) \\
31.41 (8.45) \\
30.00 (9.98) \\
30.55 (10.46)
\end{tabular}
& \begin{tabular}[c]{@{}c@{}}
1.25 (1.59) \\
1.03 (0.94) \\
0.82 (1.32) \\
0.66 (1.01)
\end{tabular}
\\
\hline

\multirow{3}{*}{Yes}
& 4
& \begin{tabular}[c]{@{}c@{}}
No \\
$\gamma$ \\
$(\alpha,\beta)$ \\
$\gamma,\;(\alpha,\beta)$
\end{tabular}
& \begin{tabular}[c]{@{}c@{}}
67.56 (9.44) \\
67.47 (10.57) \\
69.61 (9.79) \\
67.98 (10.62)
\end{tabular}
& \begin{tabular}[c]{@{}c@{}}
30.98 (8.73) \\
31.52 (8.93) \\
28.97 (9.07) \\
30.99 (9.97)
\end{tabular}
& \begin{tabular}[c]{@{}c@{}}
1.47 (1.66) \\
1.01 (0.76) \\
1.52 (1.70) \\
1.03 (0.91)
\end{tabular}
\\ \cline{2-6}

& 8
& \begin{tabular}[c]{@{}c@{}}
No \\
$\gamma$ \\
$(\alpha,\beta)$ \\
$\gamma,\;(\alpha,\beta)$
\end{tabular}
& \begin{tabular}[c]{@{}c@{}}
67.45 (9.28) \\
67.62 (10.77) \\
70.50 (10.48) \\
69.84 (11.12)
\end{tabular}
& \begin{tabular}[c]{@{}c@{}}
31.04 (8.54) \\
31.36 (10.06) \\
28.03 (9.69) \\
28.93 (10.48)
\end{tabular}
& \begin{tabular}[c]{@{}c@{}}
1.51 (1.69) \\
1.03 (0.79) \\
1.47 (1.56) \\
1.23 (1.13)
\end{tabular}
\\ \cline{2-6}

& 17
& \begin{tabular}[c]{@{}c@{}}
No \\
$\gamma$ \\
$(\alpha,\beta)$ \\
$\gamma,\;(\alpha,\beta)$
\end{tabular}
& \begin{tabular}[c]{@{}c@{}}
67.54 (9.58) \\
67.65 (10.98) \\
71.18 (10.65) \\
70.89 (11.31)
\end{tabular}
& \begin{tabular}[c]{@{}c@{}}
31.02 (8.89) \\
31.36 (10.33) \\
27.37 (10.06) \\
27.77 (10.60)
\end{tabular}
& \begin{tabular}[c]{@{}c@{}}
1.44 (1.67) \\
0.99 (0.74) \\
1.45 (1.72) \\
1.34 (1.28)
\end{tabular}
\\
\hline
\end{tabular}
\end{table}